\documentclass[12pt]{article}

\usepackage[english]{babel}
\usepackage[utf8]{inputenc}
\usepackage{amsmath,amsthm,amsfonts,amscd,amssymb,eucal,latexsym,slashed}
\usepackage{cite}
\usepackage{epsfig}
\usepackage{graphicx}
\usepackage{cite}

\newtheorem{theorem}{Theorem}[section]

\newtheorem{lemma}[theorem]{Lemma}
\newtheorem{proposition}[theorem]{Proposition}

\newcommand{\ie}{{\it i.e.\ }}
\newcommand{\eg}{{\it e.g.\ }}

\newcommand{\viz}{{\it viz.\ }}

\def\RR{{\mathbb R}}

\def\CC{{\mathbb C}}
\def\NN{{\mathbb N}}
\def\TT{{\mathbb T}}

\def\cA{{\cal A}}
\def\cB{{\cal B}}
\def\cC{{\cal C}}

\def\cH{{\cal H}}

\def\cM{{\cal M}}
\def\cN{{\cal N}}
\def\cO{{\cal O}}

\def\bcA{\mbox{\boldmath $\cA$}}
\def\bcC{\mbox{\boldmath $\cC$}}

\def\bcN{\mbox{\boldmath $\cN$}}
\def\sbcN{\mbox{\boldmath $\scriptstyle \cN$}}

\def\bomegaN{\mbox{\boldmath $\omega_{\cN}$}}

\newcommand{\rest}{\upharpoonright}

\newcommand{\Ad}[1]{\mbox{Ad} \hspace{1pt} #1}

\title{Superposition, transition probabilities 
and \\ primitive observables in infinite quantum systems}

\author{{\Large Detlev Buchholz\,${}^a$
\ and \  Erling St{\o}rmer\,$^b$} \\[2mm]
${}^a$ Institut f\"ur Theoretische Physik,
Universit\"at G\"ottingen, \\ 
37077 G\"ottingen, Germany  \\[1mm]
${}^b$   Department of Mathematics, University of Oslo, \\
0316 Oslo, Norway }

\date{}

\begin{document}

\maketitle

\begin{abstract}
\noindent The concepts of 
superposition and of transition probability, 
familiar from pure states in quantum physics,
are extended to locally normal states on funnels of type~I$_\infty$ 
factors. Such funnels
are used in the description of infinite systems,  
appearing for example in quantum field theory or in quantum statistical 
mechanics; their respective constituents  
are interpreted as algebras of observables localized 
in an increasing family of nested spacetime regions.    
Given a generic reference state (expectation functional) on a funnel, \eg  
a ground state or a thermal equilibrium state, it is shown that 
irrespective of the global type of this state all of its 
excitations, generated by the adjoint action of elements of the funnel, 
can coherently be superimposed in a 
meaningful manner. Moreover, these states are 
the extreme points of their convex hull and as such are analogues 
of pure states. As further support of 
this analogy, transition probabilities are defined,
complete families of orthogonal states are exhibited 
and a one--to--one correspondence between the states and 
families of minimal projections on a Hilbert space is established. 
The physical interpretation of these quantities relies 
on a concept of primitive observables. It 
extends the familiar framework  of 
observable  algebras and avoids some counter intuitive features of that 
setting. Primitive observables admit a consistent 
statistical interpretation of corresponding measurements
and their impact on states is described by a variant of 
the von Neumann--L\"uders projection postulate.
\end{abstract}

\newpage

\section{Introduction} 
\setcounter{equation}{0}
Local quantum physics \cite{Ha}, as opposed to quantum mechanics, 
incorporates the idea that one can assign observables to  
bounded spacetime regions where corresponding measurements 
can be carried out. This point of view, relying on the Heisenberg
picture, has proved to be fruitful in the analysis of states 
in systems with an infinite number of degrees of freedom, appearing for 
example in quantum field theory or in quantum statistical mechanics. 
The global properties of these states can be quite different from 
those in standard quantum mechanics. Their description 
often requires algebras which are distinct from the familiar algebra 
$\cB(\cH)$ of all bounded operators on some separable Hilbert space $\cH$, 
being the prototype of a factor of type~I$_\infty$ according 
to the classification 
of von Neumann algebras. In fact, one also encounters algebras of 
type~II$_\infty$ and III. For states on the latter algebras the
celebrated superposition principle fails and there exists up to now no 
operationally meaningful definition of transition probabilities, 
such as for pure states. 

The restrictions of these global 
states (expectation functionals) to observables in bounded 
spacetime regions, however, behave generically like states with a 
limited number of degrees of freedom \cite{BuWi}. So a description of these 
partial states in terms of type~I$_\infty$ factors is meaningful and 
also instructive \cite{DoLo}. This insight triggered studies of funnels of 
type~I$_\infty$ factors $\cN_n \subset \cN_{n + 1}$ 
with common identity, $n \in \NN$,    
which are interpreted as algebras of observables associated with 
an increasing family (net) of nested spacetime 
regions \cite{HaKaKa,Ta}.   
Thinking of a net of strictly increasing regions it
is natural to assume that the algebras $\cN_n$ are also strictly 
increasing with increasing $n$. This feature can be expressed by 
the condition that the subalgebra of operators in 
$\cN_{n+1}$ which commute with all operators 
in $\cN_n$, in notation $\cN_n^{\, \prime} \bigcap \cN_{n+1}$, 
has infinite dimension and hence is also a type~I$_\infty$ factor,
$n \in \NN$.  Such funnels were named 
``proper sequential type~I$_\infty$ funnels'' by Takesaki \cite{Ta}; 
since we restrict attention here to this case 
we will use the shorter term funnel for them. 
We denote the algebra generated by a given funnel by
$\bcN \doteq \bigcup_n \cN_n$; it can be interpreted
as the algebra generated by all observables which are localized 
in bounded spacetime regions. 

The physical states on a funnel are described by positive, linear 
and normalized expectation functionals $\, \omega : \bcN \rightarrow \CC$
which are locally normal, \viz weak operator continuous
on the unit ball of each subalgebra $\cN_n$, $n \in \NN$. 
Given such a state $\omega$ one obtains 
by the Gelfand--Naimark--Segal (GNS) construction
a faithful representation of  
the funnel on some separable Hilbert 
space $\cH$; hence one can identify $\bcN$ with a subalgebra
of $\cB(\cH)$. Moreover, there is a  
unit vector $\Omega \in \cH$ such that 
$\omega(A) = \langle \Omega, A \Omega \rangle$, $A \in \bcN$, 
and $\Omega$ is cyclic for $\bcN$, \viz  
the subspace~$\bcN \, \Omega$ is dense in~$\cH$. 
Depending on the choice of state, the closure of  
$\bcN \subset \cB(\cH)$ in the weak operator
topology, denoted by $\cM$, can be of any infinite type. Yet 
this fact is of no relevance here. 

We will restrict attention to states $\omega$
whose representing vector $\Omega \in \cH$ is separating for~$\bcN$, 
\viz the equality $A \Omega = 0$ for $A \in \bcN$
implies $A = 0$. Hence no observable in 
$\bcN$ has a sharp (non--fluctuating) value in the state, so one does not
have any \textit{a priori} information about 
local properties of the corresponding ensemble.
This feature is implied by the more stringent 
condition that~$\Omega$ is cyclic for each  
algebra $\cN_n^{\, \prime} \bigcap \cN_{n+1}$, $n \in \NN$.
We shall say that $\omega$ is a \textit{generic state} on $\bcN$ 
if its GNS--vector $\Omega$ has the latter property. Indeed,
it follows from a result of Dixmier and Marechall \cite{DiMa}
that almost all rays $\CC \, \Omega \subset \cH$ arise from  
generic  states
$\omega$ (they form a $G_\delta$ set which is dense). 

Generic states appear also frequently in physics,  
prominent examples being states of finite energy \cite[Lem.\ 5]{Bo}, 
thermal equilibrium states in quantum field theory \cite{Ja}
and states on curved spacetimes satisfying a
microlocal spectrum condition  \cite{StVeWo}. They 
may be regarded as some \textit{ad hoc} description of 
a global background (``state of the world'') in which 
local operations and measurements are performed.
Given a generic state $\omega$ on $\bcN$ we consider 
its local excitations \  
$\bomegaN \doteq 
\{ \omega_A = \omega \circ \Ad{A} : A \in \bcN, \, \omega_A(1) = 1 \}$, 
where $\Ad{A}$ denotes the adjoint action of $A$ on~$\bcN$ given
by $\Ad{A} \, (B) \doteq A^* B A$, $B \in \bcN$. 
The completion of the convex hull of 
$\bomegaN$ in the norm topology induced by $\bcN$,
called the folium of~$\omega$, coincides with the set
of normal states on~$\cM$ \cite{HaKaKa}. But we will not
deal with this completion here. 

In the present investigation we analyze for given generic state 
$\omega$ on $\bcN$ the structure of its local excitations $\bomegaN$
and show that they have many properties in common with the set of 
pure states in finite quantum systems. 
In the subsequent section we prove that there exists a 
canonical (bijective) lift from $\bomegaN$ to 
``normalized'' rays in $\bcN$
given by $\omega_A \mapsto \TT A$,   
$\TT$ being the group of phase factors. (We will use the term
``ray'' also for these normalized sections.)
Since $\bcN$ is in a natural way a 
vector space equipped with a scalar product induced by $\omega$,
this operation is analogous to lifting pure states 
to rays of Hilbert space vectors. Thus the 
states in~$\bomegaN$ can coherently be superimposed
also in those cases where the underlying state~$\omega$ 
on~$\bcN$ is not pure and the weak closure $\cM$ of $\bcN$ is of any type. 
Moreover, the states $\bomegaN$ are the extreme points of
their (algebraic) convex hull, similarly to the case of finite
quantum systems, where pure states are by definition the extreme  
points of their convex hull. 

Making use of these facts we introduce in Sect.~3 an intrinsic 
concept of transition probability between pairs of states 
$\omega_A, \omega_B \in \bomegaN$,
putting $\omega_A \cdot \omega_B \doteq |\omega(A^* B)|^2$. 
This product is locally continuous in both entries 
and can be extended to the convex hull of~$\bomegaN$. 
The states are said to be orthogonal if $\omega_A \cdot \omega_B = 0$
and there exist families of mutually orthogonal states
$\omega_{A_n} \in \bomegaN$, $n \in \NN$, satisfying the 
completeness relation  
$\sum_n \omega_{A_n} \cdot \omega_B = 1$ for all~$\omega_B \in \bomegaN$.

In analogy to the relation between pure states and one--dimensional 
projections whose linear span forms a *--algebra
of compact operators, 
we show in Sect.\ 4 that the linear span of~$\bomegaN$, denoted by
$\mbox{Span} \, \bomegaN$, carries a product \ 
$\mbox{Span} \, \bomegaN \times \mbox{Span} \, \bomegaN 
\rightarrow \mbox{Span} \,\bomegaN$
and a \mbox{*--operation}. They turn $\mbox{Span} \, \bomegaN$
into a *--algebra which
is isomorphic to a subalgebra of the compact operators.
There is also a spectral theorem for $\mbox{Span} \, \bomegaN$
which implies that all states in the convex hull of
$\bomegaN$ can be decomposed into convex 
combinations of orthogonal states. 

In Sect.\ 5
we discuss the effect of inner operations on the states 
$\omega_A \in \bomegaN$ and propose a concept of ``primitive observables''. 
A primitive observable is fixed by specifying 
the adjoint action $\Ad{U}$ of some unitary $U \in \bcN$ which induces the map  
$\omega_A \mapsto \omega_A \circ \Ad{U} = \omega_{\,U \! A}$. 
It is conceptually important 
that the resulting transition probabilities 
$\omega_A \cdot \omega_{\, U \! A} = |\omega_A(U)|^2$
between the initial and final states
may be regarded as observable since the unitaries $U$ are normal 
operators and can therefore be decomposed into two commuting observables
whose mean values can in principle be determined experimentally. 
Hence the above formula for transition probabilities
is a physically meaningful extension of the corresponding one 
for pure states, involving one--dimensional projections.
We therefore assert that the result of a measurement of a primitive 
observable in a given ensemble corresponds to the transition probability 
fixed by the corresponding operation in the given state.

We will show that one can recover from primitive observables 
the familiar observables and their expectation values  
by tuning the underlying unitaries and proceeding to a suitable 
limit. But the primitive observables also provide tools
for the analysis of states which are not available in the 
conventional framework of observable algebras. For example, 
they allow to evade certain counter intuitive 
features and apparent paradoxa in the interpretation 
of relativistic quantum field theories originating from 
the Reeh--Schlieder property of the vacuum   \cite{Hal,BuYn}. Moreover,
the primitive observables lead naturally to   
a concept of commensurability which 
generalizes the condition of commutativity for 
commensurable  observables. This generalization could be 
useful in the discussion of causality properties 
of theories which do not comply with the standard 
postulate of locality of observables \cite{Ha}. 
The article concludes with a brief summary and outlook. 

\section{Superpositions and mixtures}
\setcounter{equation}{0}

We establish in this section the asserted 
extension of the superposition principle to the local excitations
$\bomegaN$ of a given generic state $\omega$ on the algebra 
$\bcN$ generated by a funnel $\cN_n \subset \cN_{n+1}$, $n \in \NN$.
As explained in the introduction, we may assume that 
$\bcN$ is concretely given on some Hilbert space $\cH$,
$\bcN \subset \cB(\cH)$, and that $\omega$
is a vector state induced by some 
unit vector $\Omega \in \cH$ which is cyclic for 
$\cN_n^{\, \prime} \bigcap \cN_{n+1}$, $n \in \NN$.
The following basic lemma is an easy consequence of this 
cyclicity. It is used at various points in the subsequent analysis.

\begin{lemma} \label{basic} 
Let $\omega_{A_m} \in \bomegaN$ and $c_m \in \CC$, $m = 1, \dots , M$, 
such that $\sum_{m = 1}^M c_m \, \omega_{A_m} = 0$. Then 
$\sum_{m = 1}^M c_m \, A_m^* C A_m = 0$ for every $C \in \bcN$. 
\end{lemma}
\begin{proof}
There exists some $n \in \NN$ such that  
$A_m \in \cN_n$, $m = 1, \dots , M$. Given 
any $C \in \bcN$ one may assume, choosing $n$ sufficiently 
large, that also $C \in \cN_n$. Now let 
$X,Y \in \cN_n^{\, \prime} \bigcap \cN_{n+1}$, then 
\begin{equation*}
\begin{split}
\sum_{m = 1}^M c_m \, \langle X \Omega, A_m^* C A_m \, Y \Omega \rangle 
= \sum_{m = 1}^M c_m  \,
\langle \Omega, A_m^* X^* C Y A_m \Omega \rangle
= \sum_{m = 1}^M c_m \,  \omega_{A_m}(X^*CY) = 0 \, .
\end{split}
\end{equation*}
Since this relation holds for arbitrary 
$X,Y \in \cN_{n}^{\, \prime} \bigcap \cN_{n+1}$ and 
$\Omega$ is cyclic for this algebra it follows that
$\sum_{m=1}^M c_m A_m^* C A_m = 0$, as claimed. 
\end{proof}

The second technical ingredient in our analysis is the basic fact 
that for any given $n \in \NN$ there exists a pure state 
$\omega_n$ on the type I$_\infty$ factor $\cN_{n+1}$
which is an extension of $\omega \rest \cN_n$, \viz
$\omega_n(C) = \omega(C)$, $C \in \cN_n$.
This is an immediate consequence of the 
assumption that all algebras $\cN_n$ are of 
type~I$_\infty$, implying   
$\cN_{n+1} = \cN_n \otimes (\cN_n^{\, \prime} \bigcap \cN_{n+1})$. 
Since the second tensor factor is infinite dimensional,  
any normal state (density matrix) on $\cN_n$ can be extended 
to a pure state on $\cN_{n+1}$. Having chosen an  extension $\omega_n$
of $\omega \rest \cN_n$, there is some
non--trivial minimal (``one--dimensional'') projection
$E_n \in \cN_{n+1}$ satisfying $E_{n} C E_{n} = \omega_n(C) \, E_n$,
$C \in \cN_{n+1}$. We will make repeatedly use of this result. The 
following fundamental  
proposition is based on these technical ingredients. 

\begin{proposition} \label{lift}
Let $\omega_A , \omega_B \in \bomegaN$  satisfy  
$\omega_A = \omega_B$. There is a phase factor 
$t \in \TT$ such that $B = t  A$. Conversely, if 
$B = t  A$, $t \in \TT$, then $\omega_A = \omega_B$.
Consequently, there exists a bijective lift from the states
in $\bomegaN$ to rays in $\bcN$, given by 
$\omega_A \mapsto \TT A$.
\end{proposition}
\begin{proof}
Given $\omega_A , \omega_B \in \bomegaN$ there is 
some $n \in \NN$ such that $A,B \in \cN_n$ and, 
as explained above, there 
is a minimal projection $E_n \in \cN_{n+1}$ such that 
$E_n C E_n = \omega(C) E_n$, $C \in \cN_n$. 
Moreover, since $\omega_A = \omega_B$ one infers from the basic 
lemma that  
$A^* C A = B^* C B$, $C \in \bcN$. 
Inserting here $C = B E_n$ and multiplying the resulting equality
from the left by $E_n$ yields
\mbox{$E_n \, \omega(A^*B) A = E_n \, \omega(B^* B) B$}. 
This implies  $\omega(A^* B) \, A = \omega(B^*B) \, B$
since $\Omega$ is separating for~$\cN_n$ and 
$\| E_n Z \|^2 = \| E_n Z Z^* E_n \| = \| Z^* \Omega \|^2$ for
$Z \in \cN_n$.
It then follows from the normalization condition 
$\omega(A^*A) = \omega(B^* B) = 1$ that 
$B = \omega(A^* B) A$ and $\omega(A^* B) \in \TT$.
The second statement is obvious, completing the proof. 
\end{proof}

It is an immediate  
consequence of this result that 
there holds the following version of the superposition
principle in the present framework. \\[2mm]
\textbf{Definition:} 
Let $A,B \in \bcN$ and 
$c_A, c_B \in \CC$ such that $\omega((c_A A + c_B B)^* (c_A A + c_B B)) = 1$. 
There is a unique state $\omega_{c_A A + c_B B} \in \bomegaN$
corresponding to the ray $\TT (c_A A + c_B B)$.
It defines a coherent superposition of the states  
$\omega_A, \omega_B \in \bomegaN$.  \\[2mm]
\indent 
In the next step we establish continuity 
properties of the lifts from $\bomegaN$ to rays in $\bcN$.
To this end we equip $\bomegaN$ with the norm topology 
induced by $\bcN$, \ie the norm distance between states
$\omega_A, \omega_B \in \bomegaN$ is defined by \ 
$\| \omega_A - \omega_B \| \doteq 
\sup_{C \in \sbcN} |\omega_A(C) - \omega_B(C)|/\|C\|$. 

\begin{proposition} \label{continuity}
Let $n \in \NN$ be fixed and let $A_m \in \cN_n$, $m \in \NN$, be 
a uniformly bounded sequence of operators such that 
$\omega_{A_m}$, $m \in \NN$, form a Cauchy sequence 
in $\bomegaN$. There exists a sequence $t_m \in \TT$,  $m \in \NN$,
such that $\,$ $t_m A_m$,  $m \in \NN$, converges to some operator 
$A \in \cN_n$ in the strong operator topology, and
the sequence $\omega_{A_m}$, $m \in \NN$, 
converges to $\omega_A \in \bomegaN$ in norm.
\end{proposition}
\begin{proof}
One makes use again of the basic fact that there exists a minimal
projection $E_n \in \cN_{n+1}$ satisfying \
$E_n  C E_n = \omega(C) E_n$, $C \in \cN_n$. Hence, picking
$X,Y \in \cN_n^\prime \bigcap \cN_{n+1}$, one obtains the estimate for any 
$l,m \in \NN$ 
\begin{equation*}
\begin{split}
& |\langle X \Omega, A_m^* A_m E_n A_m^* A_m \, Y \Omega \rangle 
- \langle X \Omega, A_l^* A_m E_n A_m^* A_l \, Y \Omega \rangle| \\
& = |\omega_{A_m}(X^*A_m E_n A_m^* Y) -  \omega_{A_l}(X^*A_m E_n A_m^* Y)|
\leq  \sup_k \|A_k\|^2 \, \|X\| \|Y\| \, 
 \| \omega_{A_m} - \omega_{A_l} \| \, .
\end{split}
\end{equation*}
Since $\Omega$ is cyclic for $\cN_n^{\, \prime} \bigcap \cN_{n+1}$ it follows that 
$(A_m^* A_m E_n A_m^* A_m - A_l^* A_m E_n A_m^* A_l) \rightharpoonup 
0$ in the weak operator topology as \ $l,m \rightarrow \infty$. 
Multiplying this relation from the left and right by~$E_n$ one finds
that $(E_n - |\omega(A_m^* A_l)|^2 E_n) \rightharpoonup 0$ 
and consequently $|\omega(A_m^* A_l)| \rightarrow 1$ in this limit.
Moreover, since $\|  (A_l -  \omega(A_m^* A_l) A_m) \Omega \|^2
= (1 - |\omega(A_m^* A_l)|^2 )$ one obtains, taking 
scalar products between  $A_k \Omega$ and 
the vector under the norm,
$|\omega(A_k^* A_l) - \omega(A_m^* A_l)  \omega(A_k^* A_m) | 
\rightarrow 0 \, , $ uniformly in $k \in \NN$. Since $\TT$ is 
compact it follows from these facts 
that there exist $t_l, t_m \in \TT$ such that 
$| \omega(A_m^*A_l) - t_m t_l^{-1} | \rightarrow 0$ and,
making use of this information in the above vector norm, one arrives at 
 $\|  (t_l A_l -  t_m A_m) \Omega \| \rightarrow 0$
as $l,m \rightarrow \infty$. Since 
the sequence $t_m A_m \in \cN_n$, $m \in \NN$, is uniformly bounded and  
$\Omega$ is separating for
$\cN_n$ it is then clear that it converges 
in the strong operator topology and has a limit $A \in \cN_n$
since $\cN_n$ is complete. 
The remaining part of the statement follows from the simple estimate
$\| \omega_{A_m} - \omega_A \| 
\leq 2 \inf_{\, t \in \TT}  \|t A_m \Omega - A \Omega \| $. 
\end{proof}

From the point of view of physics 
this result says that the family of 
states~$\bomegaN$ is complete under the action of all
possible operations 
which can be performed in the given  
state~$\omega$ within bounded spacetime regions 
of arbitrary size.  We therefore say $\bomegaN$
is \textit{locally complete}. This set is thus    
a natural framework for the description of those states 
which can realistically be prepared in a given background.
Secondly, the above result shows that there exist
pointwise continuous sections of the lift from $\bomegaN$ 
to rays in $\bcN$. More explicitly, given any 
$\omega_A \in \bomegaN$, 
let~$\omega_{A_m }$ be any sequence of states,  
generated by uniformly bounded operators $A_m \in \cN_n$, $m \in \NN$, 
for some $n \in \NN$, which converges in norm to 
$\omega_A$. Then there is a sequence $t_m \in \TT$, $m \in \NN$, such
$t_m A_m \rightarrow A$ in the strong operator topology as 
$m \rightarrow \infty$. 
We will refer to this fact by saying the lift is  
 \textit{locally continuous} and use this term  
similarly for other functions on~$\bomegaN$. 

Next, we turn to the analysis of mixtures of 
states in $\bomegaN$. For any given $M \in \NN$, 
arbitrary states $\omega_{A_m} \in  \bomegaN$ 
and numbers $0 \leq p_m \leq 1$ summing up to $1$,
$m = 1, \dots , M$, we consider the convex 
combinations  $\sum_{m = 1}^M p_m \, \omega_{A_m}$. They form   
the algebraic convex hull $\mbox{Conv}\,\bomegaN$ of~$\bomegaN$.
In the subsequent proposition it is shown that $\bomegaN$ consists of the  
extreme points of $\mbox{Conv} \, \bomegaN$.

\begin{proposition}
Let $\omega_A \in \bomegaN$ and let 
$\sum_{m = 1}^M p_m \, \omega_{A_m} = \omega_A$, 
where $\omega_{A_m} \in \bomegaN$ and 
$p_m$ are positive numbers summing up to $1$, $m = 1, \dots , M$. Then
$\omega_{A_1} = \dots = \omega_{A_M} = \omega_A$. 
\end{proposition}
\begin{proof}
Let $n \in \NN$ be big enough such that 
(the rays of) $A, A_1, \dots , A_M \in \cN_n$
and let \mbox{$E_n \in \cN_{n+1}$} be a minimal projection satisfying
$E_n C E_n = \omega(C) E_n$, $C \in \cN_n$. 
Now the input of the statement implies,  
making use of Lemma~\ref{basic}, 
$A^* E_n A = \sum_{m = 1}^M p_m \, A_m^* E_n A_m$.
Furthermore, $A^* E_n ACA^* E_n A = \omega_n(ACA^*) A^* E_n A $,
$C \in \cN_{n+1}$,  where $\omega_n$ is the chosen extension
of $\omega \rest \cN_n$ to a pure state on $\cN_{n+1}$.
Inserting $C=1$ and noticing that \
\mbox{$\omega_n(AA^*) = \omega(AA^*) = \|A^* \Omega \|^2 \neq 0$} since 
$\Omega$ is separating, it is clear that $A^* E_n A$
is a multiple of a minimal projection in $\cN_{n+1}$.
This applies equally to the operators  
$A_m^* E_n A_m$, $m = 1, \dots , M$. Since 
$p_m > 0$, $m = 1, \dots , M$, and  a 
minimal projection cannot be decomposed into 
a sum of different positive operators, it follows that \ 
$ \omega(AA^*) \, A_m^* E_n A_m =  \omega(A_mA_m^*) \, A^* E_n A$,
$m = 1, \dots , M$.  Multiplying these equalities from the 
left by $E_n A$ and making use again of the fact that 
$\Omega$ is separating for~$\cN_n$ and that $\omega(AA^*) \neq 0$, 
one arrives at 
$\omega(A A_m^*) A_m = \omega(A_mA_m^*) A$, $m = 1, \dots , M$.
Taking also into account the normalization  condition
$ \omega(A^*A) = \omega(A_1^* A_1) = \dots = \omega(A_M^* A_M) =1$,   
one concludes that $|\omega(A A_m^*)| =  \omega(A_mA_m^*) \neq 0$,
so there exist $t_m \in \TT$ such that 
$A_m = t_m A$, $m = 1, \dots , M$. The statement then follows. 
\end{proof}

\section{Transition probabilities}
\setcounter{equation}{0}

We have seen in the previous section that 
the states in~$\bomegaN$ have many 
properties in common with pure states, irrespective of the 
type of the underlying state~$\omega$ on $\bcN$. 
Further support for this interpretation is provided by 
the existence of a meaningful concept of transition 
probabilities on $\bomegaN$. \\[1.5mm]
\textbf{Definition:} Let $\omega_A, \omega_B \in \bomegaN$.
The transition probability between these states is given 
by $\omega_A \cdot \omega_B \doteq |\omega(A^*B)|^2$.
The states are said to be orthogonal if 
$\omega_A \cdot \omega_B = 0$. 
(Note that the definition is meaningful since 
the lifts  $\omega_A \mapsto \TT A$,    
$\omega_B \mapsto \TT B$ are  injective.) 

\vspace*{1.5mm}
Another generalization of the concept of transition 
probabilities  was proposed by Uhl\-mann in \cite{Uh}.
Our present definition of $\omega_A \cdot \omega_B$
differs from 
$\omega_A \overset{\mbox{\tiny \bf U}}{\cdot} \omega_B$, 
the quantity given by Uhlmann. In fact,      
$\omega_A \cdot \omega_B \leq 
\omega_A \overset{\mbox{\tiny \bf U}}{\cdot} \omega_B$, 
where one has equality for  arbitrary states in
$\bomegaN$ only if $\omega$ is pure. 
Whereas the Uhlmann concept has proved to be 
useful in mathematics, we believe that our present definition 
is more adequate in physics. Some basic properties are compiled 
in the following proposition.

\begin{proposition}
 Let $\omega_A, \omega_B \in \bomegaN$. Then \\[1.5mm]
(i) \hspace{6pt} $0 \leq \omega_A \cdot \omega_B \leq 1$    and 
 $\omega_A \cdot \omega_B = \omega_B \cdot \omega_A$.  \\[0.5mm]
(ii)  \hspace{4pt}  $\omega_A \cdot \omega_B \leq 
1 - (1/4) \, \| \omega_A - \omega_B \|^2$; 
equality holds for all $\omega_A, \omega_B$ iff 
$\omega$ is a pure state on $\bcN$. \\[0.5mm]
(iii) \hspace{-2pt} The map
$\omega_A, \omega_B \mapsto \omega_A \cdot \omega_B$ is locally
continuous. 
\end{proposition}
\begin{proof}
Point (i) follows from the Schwarz inequality 
$0 \leq  |\omega(A^*B)|^2 \leq \omega(A^*A) \, \omega(B^*B) = 1$
and the fact that $\overline{\omega(A^*B)} = \omega(B^* A)$.
For the proof of point (ii) one makes use of the fact 
that the states in $\bomegaN$ can be extended to $\cB(\cH)$
in the GNS representation induced by $\omega$ and defines   
$\| \omega_A - \omega_B \|_{\cB(\cH} 
\doteq \sup_{X \in \cB(\cH)} |\langle A \Omega, X A \Omega \rangle 
- \langle B \Omega, X B \Omega \rangle|/ \| X \| $.
Since $\bcN \subset \cB(\cH)$ it is clear that 
$\| \omega_A - \omega_B \| \leq \| \omega_A - \omega_B \|_{\cB(\cH}$;
moreover, 
$ \omega_A \cdot \omega_B = |\langle A \Omega, B \Omega \rangle|^2
= 1 - (1/4) \,  \| \omega_A - \omega_B \|_{\cB(\cH}^{\, 2} \,$,
as has been shown in \cite{RoRo}. The inequality in (ii) then follows.
If $\omega$ is a pure state on $\bcN$, whereby
the weak closure of~$\bcN$ coincides with $\cB(\cH)$, 
one has $\| \omega_A - \omega_B \| = \| \omega_A - \omega_B \|_{\cB(\cH}$,
so equality in (ii) obtains in this case.
If $\omega$ is not pure, the commutant 
$\bcN^{\, \prime} \subset \cB(\cH)$ contains some non--trivial 
unitary operator~$V$. Since 
$\Omega$ is cyclic for $\bcN$ (by the GNS--construction)
it is separating for $\bcN^{\, \prime}$ and consequently
$|\langle \Omega, V \Omega \rangle| < 1$.
Now let $A_m \in \bcN$, $m \in \NN$, be any sequence such that
$A_m \Omega \rightarrow V \Omega$ strongly as $m \rightarrow \infty$.
Then $\lim_{m \rightarrow \infty} \| \omega_{A_m} - \omega \| = 0$
and $\lim_{m \rightarrow \infty} |\langle \Omega, A_m \Omega \rangle | < 1$,
so one cannot have equality in (ii) for arbitrary states, as claimed. 
The remaining point (iii) is an immediate consequence of the 
fact that the lifts $\omega_A \mapsto \TT A$ and 
$\omega_B  \mapsto \TT B$ are locally continuous, cf.\ the
remark after Proposition \ref{continuity}, and this continuity property is 
passed on to the transition probabilities.  
\end{proof}

Based on the concept of transition probability for the states in $\bomegaN$  
it is easy to exhibit complete families of orthogonal states.
For, by the Gram--Schmidt algorithm, one finds orthonormal systems 
of vectors  
$A_m \Omega \in \bcN \Omega$, $m \in \NN$, which are 
complete in $\cH$ since $\bcN \Omega$ is dense. In particular,
$\sum_m |\langle B \Omega, A_m \Omega \rangle |^2 = 1$
for every $B \in \bcN$ satisfying 
$\| B \Omega \| = 1$. 
Hence the states $\omega_{A_m} \in \bomegaN$, $m \in \NN$,
form a complete family of orthogonal states in the following 
sense. \\[2mm]
\textbf{Definition:} A family of orthogonal states 
$\omega_{A_m} \in \bomegaN$, $m \in \NN$, is said to be 
complete if $\sum_m \omega_B \cdot \omega_{A_m} = 1$
for every $\omega_B \in \bomegaN$. \\[2mm]
The existence of such complete families implies that one has a consistent
statistical interpretation of the transition probabilities
for the states in $\bomegaN$. 

\section{Algebra of states}
\setcounter{equation}{0}

We continue our analysis of the states $\bomegaN$ by showing that
their linear span $\mbox{Span} \, \bomegaN$
can be equipped with an associative product and a star operation, 
turning it into a *--algebra. Making use of the polarization
identity 
$\omega(A^*  C  B) =
(1/4) \, \sum_{j=0}^3 i^j \, \omega((A + i^j B)^* \, C \, (A + i^j B))$ 
for $A,B,C \in \bcN$, it is apparent that 
the functional $\omega(A^* \, \cdot \, B)$ on $\bcN$ 
is contained in $\mbox{Span} \, \bomegaN$. 
Whereas this functional depends on the 
operators \mbox{$A, B \in \bcN$}, the 
functional  in the subsequent definition does not depend on 
the specific choice of operators from the 
rays $\TT A$ and~$\TT B$, respectively. It is therefore well defined
on $\bomegaN \times \bomegaN$. \\[2mm] 
\textbf{Definition:} Let $\omega_A, \omega_B \in \bomegaN$. The
functional $\omega_A \times \omega_B  \in 
\mbox{Span} \, \bomegaN$ is defined by
$$  \omega_A \times \omega_B (C) =  \omega(A^*B) \, \omega(B^* C A) \, ,
\quad C \in \bcN \, . 
$$
Note that $\omega_A \times \omega_B (1) = \omega_A \cdot \omega_B$.  

\vspace*{2mm} In the subsequent proposition we show that the map
$\omega_A, \omega_B \mapsto \omega_A \times \omega_B$
extends linearly in both entries to an associative product
on $\mbox{Span} \, \bomegaN$. We also define an 
antilinear involution~$\dagger$ on this space. 
\begin{proposition} \label{algebra}
The map $\omega_A, \omega_B \mapsto \omega_A \times \omega_B$
from $\bomegaN \times \bomegaN$ to $\mbox{Span} \ \bomegaN$
extends linearly in both entries to an associative product
on $\mbox{Span} \, \bomegaN$,  given by 
\begin{equation} \label{product}
(\sum_{k=1}^K c_k \, \omega_{A_k}) \times (\sum_{l=1}^L d_l \, \omega_{B_l})
\doteq \sum_{k=1}^K \sum_{l=1}^L  c_k  d_l \  
\omega_{A_k} \times  \omega_{B_l} \, . 
\end{equation}
Moreover, the antilinear involution $\dagger$ on $\mbox{Span} \, \bomegaN$, 
given by 
\begin{equation} \label{adjoint} 
\big(\sum_{k=1}^K c_k \, \omega_{A_k} \big)^\dagger(C)
\doteq \sum_{k=1}^K \overline{c}_k \, \omega_{A_k}(C) \, , 
\quad C \in \bcN \, ,
\end{equation}
is algebraically compatible with this product. Equipped with these 
operations, $\mbox{Span} \, \bomegaN$ becomes a *--algebra,
denoted by $\bcC$. 
\end{proposition}
\begin{proof}
For the proof that the definition (\ref{product}) is consistent 
we only need to show that the right hand side of this equation  
vanishes whenever either one of the sums on the left hand side vanishes. 
So let $\sum_{k=1}^K c_k \, \omega_{A_k} = 0$ and hence, by 
Lemma \ref{basic}, 
$\sum_{k=1}^K c_k \, A_k^* C A_k = 0$ for every $C \in \bcN$.
Given any functional $\sum_{l=1}^L d_l \, \omega_{B_l}$
and operator $C \in \bcN$, there is some $n \in \NN$ such that 
(the rays of) $A_k, B_l, C \in \cN_n$
for $k = 1, \dots , K$, $l = 1, \dots , L$. 
Let $E_n \in \cN_{n+1}$ be a non--trivial minimal projection
such that $E_n C E_n = \omega(C) E_N$, $C \in \cN_n$, and let 
$\tau_n$ be the standard semifinite trace
on $\cN_{n+1}$. Replacing in the above equation
$C$ by $\sum_{l=1}^L d_l \, B_l E_n B_l^* \, C \in \cN_{n+1}$ 
one gets 
$\sum_{k=1}^K \sum_{l=1}^L c_k d_l \, 
A_k^* B_l E_n B_l^* C  A_k = 0$ and multiplying the latter equation from the 
left and right by $E_n$ and evaluating its trace $\tau_n$ one
arrives at 
\begin{equation*}
\begin{split}
 0  & =  \sum_{k=1}^K \sum_{l=1}^L c_k d_l \, 
\tau_n (E_n A_k^* B_l E_n B_l^* C  A_k E_n) \\
&  =   \sum_{k=1}^K \sum_{l=1}^L c_k d_l \, 
\omega(A_k^* B_l)  \omega(B_l^* C A_k) \, \tau_n(E_n)  
  = \sum_{k=1}^K \sum_{l=1}^L  c_k  d_l \,   
\omega_{A_k} \times  \omega_{B_l} (C) \, . 
\end{split}
\end{equation*}
Since $C \in \bcN$ was arbitrary, this shows that the right hand side
of relation 
(\ref{product}) vanishes whenever the first sum on the left hand side
vanishes, and an analogous argument leads to the same conclusion 
if the second sum vanishes. It is then obvious that the 
product is linear in  both entries. 
The proof that the product is associative 
is a consequence of the relations for
$A,B,C,D \in \bcN$ 
\begin{equation*}
\begin{split}
& \big( \omega(A^*B)  \, \omega(B^* \, \cdot \, A)
\times \omega_C \big)(D) = \omega(A^*B) \, \omega(B^* C) \, 
\omega(C^* D A) \, , \\
& \big( \omega_A \times   \omega(B^*C)   \, \omega(C^* \, \cdot \, B) \big)(C)
= \omega(B^*C) \, \omega(A^*B) \, \omega(C^* D A) \, , 
\end{split}
\end{equation*}
which follow from the definition of the product and the polarization identity
for the functionals 
$\omega(B^* \, \cdot \, A)$ and $\omega(C^* \, \cdot \, B)$.
Finally, relation (\ref{adjoint}) and 
the polarization identity imply for  
$A,B,C \in \bcN$, $c,d \in \CC$, 
\begin{equation*}
\begin{split}
& \big(c \, \omega_A \times d \, \omega_B \big)^\dagger(C) 
=  (1/4) \sum_{j = 0}^3 \overline{cd \, \omega(A^*B) \, i^j } \, 
\omega((B + i^j A)^* C (B + i^j A)) \\
& = (\overline{cd} /4) \sum_{j = 0}^3 \omega(B^*A) \, (-i)^j  \, 
\omega((A + (-i)^j B)^* C (A + (-i)^j B))
= \big(\overline{d} \, \omega_B \times \overline{c} \, \omega_A \big)(C) \, ,
\end{split}
\end{equation*}
proving the algebraic compatibility of the antilinear involution
$\dagger$ with the product $\times$ on 
$\mbox{Span} \ \bomegaN$. The conclusion then follows. 
\end{proof}

\vspace*{2mm}
Since the states  
$\omega_A \in \bomegaN$ satisfy $\omega_A \times \omega_A = \omega_A$
and $\omega_A^{\ \dagger} = \omega_A$, they correspond to 
symmetric projections in $\bcC$. Moreover, they 
satisfy $\omega_A \times \omega_C \times \omega_A = 
(\omega_A \cdot \omega_C) \ \omega_A$,~$\omega_C \in \bomegaN$, 
and hence are minimal projections.  It is also 
clear that $\omega_A \times \omega_B = 0$ 
iff $\omega_A, \omega_B$ are orthogonal states. 
We will show next that any 
symmetric element of  $\bcC$ can be
decomposed into a sum of orthogonal minimal projections with real
coefficients (``spectral theorem'').
\begin{proposition} \label{spectraltheorem}
Let $\psi = \psi^\dagger \in \bcC$.
There exist mutually orthogonal states $\omega_{A_m} \in \bomegaN$ 
and coefficients $r_m \in \RR$, $m= 1, \dots , M$,
such that $\psi = \sum_{m=1}^M r_m \, \omega_{A_m}$. 
If $\psi \in \mbox{Conv} \ \bomegaN \subset \bcC$, then 
$r_m \geq 0$, $m= 1, \dots , M$, and \ $\sum_{m=1}^M r_m = 1$;
hence every ``mixed state'' in $\mbox{Conv} \ \bomegaN$ 
can be decomposed into a convex combination of orthogonal 
``pure states'' in $\bomegaN$.
\end{proposition}
\begin{proof} If $\psi = 0$ the statement holds trivially. So
let $\psi = \sum_{m = 1}^M c_m \, \omega_{B_m}  = \psi^\dagger$ where,
without loss of generality, the projections
$\omega_{B_m} \in \bcC$ are linearly independent and 
the coefficients $c_m$ are real and different
from $0$ for $m = 1, \dots , M$. 
A given $\psi$ has in general different decompositions of
this type. Fixing any one, one 
considers the one--dimensional projections
$B_m E_\Omega B_m^* \in \cB(\cH)$, \mbox{$m = 1, \dots , M$}, 
where $E_\Omega$ is the projection onto the ray 
$\CC \, \Omega \in \cH$. 
Denoting by 
$\tau$ the trace on~$\cB(\cH)$ one obtains 
$\tau\big((\sum_{m=1}^M  d_m \, B_m E_\Omega B_m^*) \, C \big)
=  \sum_{m=1}^M  d_m \, \omega_{B_m}(C)$ for $C \in \bcN$
and $d_m \in \CC$, $m = 1, \dots , M$. 
Thus the projections $B_m E_\Omega B_m^*$, $m = 1, \dots , M$, 
are linearly independent as 
well. Moreover, putting 
$\Psi \doteq \sum_{m=1}^M c_m \, B_m E_\Omega B_m^* \in \cB(\cH)$, 
one recovers  the given functional, 
$\tau\big( \Psi \, C \big) = \psi(C)$,  $C \in \bcN$.

The vectors $B_m\Omega \in \cH$, $m = 1, \dots , M$, 
span an $M$--dimensional subspace $\cH_M \subset \cH$ which is 
stable under the action of the self--adjoint operator~$\Psi$. 
Hence there exists a non--singular matrix 
$S_{m \, l}$, $m,l = 1, \dots, M$, such that 
the vectors 
$\sum_{l=1}^M S_{m \, l} \,  B_l \, \Omega \in \cH_M$, $m = 1, \dots, M$, 
are orthogonal, normalized and diagonalize~$\Psi \rest \cH_M$, \ie
$(\Psi - r_m 1) \, \big(\sum_{l=1}^M S_{m \, l}   B_l \Omega \big) =  0$, 
where $r_m \in \RR$, $m = 1, \dots, M$.  
Hence $\Psi$ can be represented in the form 
$\Psi = \sum_{m=1}^M r_m \, A_m E_\Omega A_m^*$, where 
the operators 
$A_m \doteq \sum_{l=1}^M S_{m,l} \,  B_l $, $m = 1, \dots, M$,
are elements of $\bcN$. It follows that
$\tau(\Psi \, C) =  \sum_{m=1}^M r_m \, \omega_{A_m}(C)$, 
$C \in \bcN$. Since by construction 
$\omega_m \cdot \omega_l = |\langle A_m \Omega, A_l \Omega \rangle |^2
= \delta_{m,l}$, $l,m = 1, \dots, M$, this establishes the desired 
decomposition of~$\psi$. 

Finally, let $\psi = \sum_{m = 1}^M c_m \, \omega_{B_m}$, where 
$c_m \geq 0$, $m = 1, \dots, M$. 
Then 
$$
(\psi \times \omega_{A_l})(1)
= \sum_{m = 1}^M c_m  \, (\omega_{B_m} \times \omega_{A_l})(1) =
\sum_{m = 1}^M c_m  \, \omega_{B_m} \cdot \omega_{A_l} \geq 0 \, ,
\quad l=1, \dots , M \, .
$$
Making use of this information in the orthogonal decomposition
$\psi = \sum_{m=1}^M r_m \, \omega_{A_m}$, 
it follows that $r_l = (\psi \times \omega_{A_l})(1) \geq 0$,
$l = 1, \dots , M$. Since the states in $\mbox{Conv} \ \bomegaN$ are
normalized one also has $1 = \psi(1) = \sum_{m=1}^M r_m$,
completing the proof of the theorem. 
\end{proof}

We conclude this discussion of the algebraic properties 
of \ $\bcC = \mbox{Span} \ \bomegaN$ by introducing a 
left and right action of $\bcN$ on this space. \\[2mm]
\textbf{Definition:} \  
Let $\psi \in \bcC$. The left, respectively, right actions 
of $A \in \bcN$ on $\psi$ are given by
\begin{equation*}
(A \times \psi)(C) \doteq \psi(AC) \quad \mbox{and} \quad 
(\psi \times A)(C) \doteq \psi(CA) \, , \qquad  C \in \bcN \, .
\end{equation*}
With this definition $\bcC$ becomes an $\bcN$--bimodule. 

\vspace*{2mm} 
Since $\bcC$ does not contain an 
identity, the underlying algebra $\bcN$ does not
correspond to a subalgebra of~$\bcC$. Yet one can recover 
the operators in  $\bcN$ as ``weak 
limits'' of operators in~$\bcC$. The appropriate topology is 
induced by the states in $\bomegaN$  which determine 
elements of the dual space of $\bcC$, defined subsequently. 
It follows from Proposition~\ref{algebra} that this definition
is consistent. \\[2mm]
\textbf{Definition:} \ Let $\omega_A \in \bomegaN$. Its 
dual action \ $\omega_A : \bcC \rightarrow \CC$ \ is defined by 
$$  \omega_A(\psi) \doteq (\omega_A \times \psi)(1) \, ,
\quad \psi \in \bcC \, . 
$$
The dual action of the elements of 
$\, \mbox{Conv} \, \bomegaN$ on $\bcC$ is defined analogously. 

\vspace*{2mm}
Making use of   Proposition \ref{algebra} one obtains by a straightforward 
computation the basic equalities for 
any $\omega_A, \omega_B, \omega_C , \omega_D \in \bomegaN$
\begin{eqnarray*}
& (\omega_A \times \omega_B \times \omega_C)(1)
= \omega(A^* B) \, \omega(B^* C) \omega(C^* A) \, , & \\
& (\omega_A \times \omega_B \times \omega_C \times \omega_D)(1)
= \omega(A^* B) \, \omega(B^* C) \omega(C^* D) \omega(D^* A) \, , &
\end{eqnarray*}
and similarly for higher products. 
Note that the numerical values of these expressions do
not change under cyclic permutations of the operators $A,B,C,D$. 
These relations are a key ingredient in the proof of the 
following result.
\begin{proposition} \ 
(i) Let $\omega_A \in \bomegaN$, 
The map  \ $\omega_A : \bcC \rightarrow \CC$  
satisfies $\omega_A(\psi^\dagger \times \psi) \geq 0$, $\psi \in \bcC$, 
and hence is a positive linear functional on $\bcC$.  \\[1mm] 
(ii) \hspace{-2.1pt}  The GNS--representation of $\bcC$ induced by the 
underlying state $\omega$ is faithful. \\[1mm]
(iii) \hspace{-4.4pt} There exists a spatial isomorphism between the,
as in (ii) represented, algebra~$\bcC$ and 
the algebra~$\bcC_\Omega \subset \cB(\cH)$ which is generated by 
the projections $\{ A E_\Omega A^* : A \in \bcN \, , \ \| A \Omega \|= 1 \}$,
where $E_\Omega$ denotes the projection onto the ray $\CC \Omega \subset \cH$.
\end{proposition}
\begin{proof}
(i) \  Let $\psi = \sum_m^M c_m \, \omega_{B_m}  \in \bcC$. 
Then, by the first basic  equality given above, 
$$ \omega_A(\psi^\dagger \times \psi)
= \sum_{k=1}^M \sum_{l=1}^M  \overline{c}_k c_l \, 
(\omega_A \times \omega_{B_k} \times \omega_{B_l}) (1) 
= \sum_{k=1}^M \sum_{l=1}^M  \overline{c}_k c_l \, 
\omega(A^* B_k)  \omega(B_k^* B_l)  \omega(B_l^* A) \geq 0 \, ,
$$
where the asserted positivity follows from the relation 
$\overline{\omega(B_k^* A)} = \omega(A^* B_k)$
and the fact that $\omega(B_k^* B_l)$ is a non--negative matrix, 
$k,l = 1, \dots , M$. 

For the proof of (ii), let $0 \neq \psi \in \bcC$  
such that $\omega(\omega_A \times \psi \times \omega_B) = 0$
for any choice of states $\omega_A, \omega_B \in \bomegaN$.
Taking the complex conjugate of this equality gives 
$\omega(\omega_B \times \psi^\dagger \times \omega_A) = 0$
for $\omega_A, \omega_B \in \bomegaN$. Hence it suffices 
to consider the case $\psi^\dagger = \psi$. According to 
Proposition \ref{spectraltheorem}, the given $\psi$ can then be brought 
into the form $\psi = \sum_{m=1}^M r_m \, \omega_{A_m}$ where
the states $\omega_{A_m}$ are mutually orthogonal and $r_m \neq 0$,
$m = 1, \dots, M $. Making use of the second basic equality  
given above, one gets
$$
0 = \omega(\omega_A \times \psi \times \omega_A) =
\sum_{m=1}^M \, r_m (\omega \times \omega_A \times \omega_{A_m} \times \omega_A)(1)
= \sum_{m=1}^M r_m \, |\omega(A)|^2 \, |\omega(A_m^* A)|^2 \, .
$$ 
If $\omega(A_l) \neq 0$ for some $l \in \{1, \dots , M \}$ one puts
$A = A_l$ into this equality, giving $r_l = 0$,  in conflict with the
assumption that $\psi \neq 0$. If $\omega(A_m) = 0$, $m = 1, \dots , M$,
one puts $A = c 1 + A_l$ for some~$c \neq 0 $ and arrives at the same 
conclusion. Thus there exists no non--zero operator $\psi \in \bcC$
which vanishes in the GNS--representation induced by $\omega$,
hence this representation is faithful. 
But note that $\omega$ is not a faithful state on $\bcC$. 

As to the proof of (iii), recall that the GNS representation of $\bcN$
induced by $\omega$ consists of two ingredients: 
(a) a vector space $\cH_\omega$ spanned by the equivalence classes 
$| \psi \rangle$ of elements of~$\bcC$,   
modulo the left ideal which 
annihilates $\omega$, which is equipped with a scalar product 
$\langle \psi_1 | \psi_2 \rangle \doteq \omega(\psi_1^\dagger \times \psi_2)$;
\ (b) a homomorphism $\pi$ from $\bcN$ to bounded operators acting
on $\cH_\omega$, given by 
$\pi(\psi_1) \, | \psi_2 \rangle \doteq | \psi_1 \times \psi_2 \rangle $,
$\psi, \psi_1, \psi_2 \in \bcN$. After these preparations it is 
straightforward to prove that the mapping
$W : \cH_\omega \rightarrow \cH$ defined by 
$$
W \, | \sum_{m = 1}^M c_m \, \omega_{A_m} \rangle \doteq
\sum_{m = 1}^M c_m \, \omega(A_m^*)  \, A_m \Omega \, ,
\qquad \sum_{m = 1}^M c_m \, \omega_{A_m}  \in \bcC \, ,
$$
is an isomorphism with dense domain in $\cH_\omega$ 
and dense range in $\cH$. In fact,
\begin{equation*}
\begin{split}
& \langle  \sum_{l = 1}^M c_l \, 
\omega_{A_l} | \sum_{m = 1}^M  \, c_m \omega_{A_m} \rangle 
 =  \sum_{l=1}^M \sum_{m = 1}^M  \overline{c}_l c_m 
\, \omega(\omega_{A_l} \times  \omega_{A_m}) \\
& =
\sum_{l=1}^M \sum_{m = 1}^M  \overline{c}_l c_m \, 
\omega(A_l) \omega(A_l^* A_m) \omega(A_m^*) 
=  \sum_{l=1}^M \sum_{m = 1}^M  \overline{c}_l c_m \, 
\omega(A_l)  \omega(A_m^*) \, \langle A_l \Omega,  A_m \Omega \rangle \, .
\end{split}
\end{equation*}
By a similar computation one obtains for arbitrary 
$\omega_A, \omega_B \in \bomegaN$ 
\begin{equation*}
\begin{split}
\langle  & \omega_A|  \, \pi(\sum_{m=1}^M c _m  \, \omega_{A_m}) \, | 
\omega_B \rangle \\
& = \omega(\omega_A \times (\sum_{m=1}^M c _m \, \omega_{A_m}) \times \omega_B)  
 = \sum_{m=1}^M c_m  \, \omega(A) \omega(A^* A_m) \omega(A_m^* B) \omega(B^*)  \\
& =  \omega(A)  \, \langle A \Omega, \big( \sum_{m=1}^M c_m
A_m E_\Omega A_m^* \big) B \Omega \rangle \,  \omega(B^*)
= \langle \omega_A| W^*  \big( \sum_{m=1}^M c_m
A_m E_\Omega A_m^* \big) \, W | \omega_B \rangle \, .
\end{split}
\end{equation*}
By comparison of the left and right hand side of this  equality 
one obtains 
$$
W \, \pi(\sum_{m=1}^M c _m  \, \omega_{A_m}) =
\big( \sum_{m=1}^M c_m \, A_m E_\Omega A_m^* \big) \, W \, , 
\qquad \sum_{m=1}^M c _m  \, \omega_{A_m} \in \bcC \, .
$$
Hence $W$ establishes a spatial isomorphism between $\pi(\bcC)$ and
$\bcC_\Omega$, completing the proof of the proposition.   
\end{proof}

The preceding proposition establishes a simple universal picture of the 
space $\bcC = \mbox{Span} \, \bomegaN$,  
spanned by the local excitations of any generic state $\omega$ 
on the funnel $\bcN$, which does not depend on the global type of $\omega \,$: 
the space may be identified with the bimodule obtained
by left and right (product) action of $\bcN$ on 
the projection~$E_\Omega$, where $\Omega \in \cH$ is the chosen vector  
representing~$\omega$. Thus it corresponds to a specific subspace  
$\bcC_\Omega \subset \cB(\cH)$ of finite rank (hence trace class)
operators. Moreover,
the transition probabilities of the states in $\bomegaN$,
defined above, can be expressed in terms of $\bcC_\Omega$
by the familiar dual action of trace class operators 
onto themselves under the trace of~$\cB(\cH)$. 
However, since $\bcC_\Omega \bigcap \bcN = \{0\}$ (as $\Omega$ is separating
for~$\bcN$), the elements of~$\bcC_\Omega$ may in general not be regarded
as genuine observables, in contrast to quantum mechanics, where
trace class operators are part of the observable algebra.     
The physical interpretation of these quantities therefore requires some
explanations which will be provided in the subsequent section. 

Let us mention in conclusion that part of the preceding mathematical results
could have been also established by 
making use of the well--known fact that 
there exist isomorphisms 
(universal localizing maps \cite{BuDoLo}) 
$\phi_n : \cB(\cH) \rightarrow \cN_{n+1}$, which leave $\cN_n$
pointwise fixed, $n \in \NN$. Yet our present approach  
reveals more closely the intrinsic nature of the proposed concepts. 

\section{Primitive observables}
\setcounter{equation}{0}

Having exhibited the mathematical similarities between the
local excitations of generic states in infinite quantum systems 
and pure states in quantum mechanics, let us turn now to the
discussion of the physical significance of this observation. 
What is missing so far is an argument that 
the transition probabilities, as defined above, are in principle 
accessible to observations. We will establish such a link in this section 
by making use of the concept of operations. 

The effect of non--mixing operations 
on states are generally described by the formula~\mbox{\cite{HaKa, HeKr}}
$\omega \mapsto (1 /\omega(V^*V)) \  \omega \circ \Ad{V}$,
\ie the operations are identified with the adjoint 
action of arbitrary elements $V \in \bcN$ on the states, followed by 
a normalization. In this generality, one incorporates 
operations where parts of the resulting ensemble are 
discarded by the observer and the remainder is considered as a 
new ensemble, as in the von Neumann--L\"uders projection
postulate.
Yet such ``state reducing operations'' require a re--normalization and   
induce highly non--linear mappings on the space of states. 
We therefore restrict attention here to operations induced by
operators $V \in\bcN $ satisfying the condition $V^* V = 1$,
\ie to isometries. 

Operations described by isometries appear naturally in the context of 
physics, prominent examples being the effects of temporary 
inner perturbations of the dynamics in the Heisenberg picture. 
The special case of unitary operators $U \in \bcN$ is of particular
interest here since unitary operators, being normal, are linear 
combinations of commuting selfadjoint operators,
\viz $U = (1/2)(U + U^*) + (i/2) \, i(U^* - U)$; they can 
therefore be regarded as observables. Moreover, any isometry 
and any projection in~$\bcN$ can be obtained as an appropriate limit 
of unitaries, as is shown in the lemma below. 
We will therefore base our physical interpretation of the mathematical 
framework on such unitaries. 

In the proof of the subsequent lemma  we make
use of the fact that every non--zero
projection $E \in \bcN$ has infinite 
dimension (its commutant in~$\bcN$ is infinite dimensional); moreover,  
there exist isometries $V \in \bcN$ with range projection~$E$, 
\viz  $V V^* = E$. 

\begin{lemma} \label{approximations} 
Let $E \in \bcN$ be any non--zero projection and let $V \in \bcN$
be any isometry with range projection $E$. Then \\[1mm]
(i) \hspace{2pt} there exists a sequence of unitaries 
$U_m \in \bcN$, $m \in \NN$,
such that $U_m \rightarrow V$ in the strong operator topology. \\[1mm]
(ii) \hspace{1pt}  there exists a sequence of isometries
$V_m \in \bcN$, $m \in \NN$, with common range projection $E$
such that $V_m  \rightharpoonup E$ in the weak operator topology
and $V_m^* \rightarrow E$ in the strong operator topology. \\[1mm]
(iii) \hspace{0pt}  for any $\omega_A \in \bomegaN$ one has \
$\sup_{\, V} |\omega_A(V)| = \omega_A(E)$, where the supremum is taken
with respect to all isometries $V \in \bcN$ with range projection $E$. 
\end{lemma}  
\begin{proof}
(i) Let $V \in \cN_n$ for some $n \in \NN$. Picking any strictly  
increasing sequence of projections $E_m \in \cN_n$ 
which converges to $1$ in the strong operator topology, 
one puts $V_m \doteq V E_m$, $m \in \NN$. Then $(1 - V_m V_m^*)$ and
$(1 - V_m^* V_m) = (1 - E_m)$ are infinite projections in $\cN_{n}$,
so there exist partial isometries $W_m \in \cN_{n}$
such that $W_m W_m^* = (1 - V_m V_m^*)$ and 
$W_m^* W_m = (1-E_m)$, $m \in \NN$. By construction
$(V_m + W_m)^* (V_m + W_m) = (V_m + W_m) (V_m + W_m)^* = 1$, 
hence these operators are unitaries, , $m \in \NN$. Moreover, 
$V_m \rightarrow V$ and $W_m \rightarrow 0$ in the strong
operator topology as $m \rightarrow \infty$, so statement (i) 
follows.  \\
(ii) Let $n \in \NN$ be sufficiently large such that 
there is a sequence of unitaries $U_m \in \cN_n$, $m \in \NN$,  
as in (i) which converges to $V$ in the strong operator topology. 
Then $V_m^* \doteq U_m V^* \rightarrow V V^* = E$
in the strong operator and $V_m = V U_m^* \rightharpoonup E$
in the weak operator topology. Since 
$V_m V_m^* = E$ and $V_m ^* V_m = 1$, statement (ii) follows. \\ 
(iii) For the proof of the remaining statement  
one first considers the case where $\omega_A$
is a faithful state on $\bcN$. Putting
$w_A \doteq \sup_{\, V} |\omega_A(V)|$, 
it follows from the preceding step, inserting  
the weakly convergent sequence of 
isometries $V U_m^*$, $m \in \NN$, into $\omega_A$,
that $w_A \geq \omega_A(E)$.
Next, let~$F$ be the weak operator limit of some sequence of
isometries with range projection~$E$ which saturates  
the bound~$w_A$. Since
$\langle A \Omega, F A \Omega \rangle =
\langle E \, A \Omega, F A \Omega \rangle $  
the bound is attained  
if $E A \Omega$ and $F A \Omega$ are parallel.
As~$A \Omega$ is separating this implies $F = c E$
and since $\| F \| \leq 1$ it follows that 
$| c | \leq 1$, whence  $w_A  \leq \omega_A(E)$. 
Since the faithful states are norm dense 
in the states~$\bomegaN$ on the funnel, this
shows that $w_A  = \omega_A(E)$ for all states $\omega_A \in \bomegaN$,  
completing the proof of the statement. 
\end{proof}

The adjoint action of a given unitary $U \in \bcN$ 
on the states, 
$\omega_A \mapsto \omega_A \circ \Ad{U} = \omega_{UA}$, 
\mbox{$\omega_A \in \bomegaN$}, 
describes the effect of a fixed physical operation
on the corresponding ensembles. The respective transition 
probabilities between the initial and final states 
are given by
\mbox{$\omega_A \cdot \omega_{\, UA} = |\omega_A(U)|^2 $.} They 
can be determined by measurements of the commuting observables 
underlying $U$ in the ensemble described by $\omega_A$. 
Simple examples illustrating this fact are the unitaries
$U_t \doteq \big( E + t \, (1-E) \big)$, where $E \in \bcN$ is 
a projection and $t \in \TT$. Then
$\omega_A \cdot \omega_{\, U_tA} = \omega_A(E)^2 + \omega_A(1-E)^2
+ 2 \, \mbox{Re}\, (t) \, \omega_A(E) \omega_A(1-E) $, so in this  
case it suffices to determine the expectation value 
of $E$ in state $\omega_A$ in order to determine the 
transition probability induced by the corresponding operation. 

Let us mention as an aside that the customary term 
``transition probability'' is slightly misleading 
in the present context since, given the initial state $\omega_A$, 
the resulting final state $\omega_{\, UA}$ is uniquely determined.  
Rather, the  quantity $\omega_A \cdot \omega_{\, UA}$ represents  
the probability of finding in the final ensemble
members of the original ensemble which ``survived'' the operation.
A more suggestive notion 
for  $\omega_A \cdot \omega_{\, UA}$ would therefore be the term 
``fidelity'', introduced  by Jozsa \cite{Jo}. It could be 
interpreted as a measure for the degree of compatibility   
of an operation with the properties of a  given initial ensemble. 
Yet in order to avoid confusion we will continue to use   
the term ``transition probability''. 
  
In order to shed further light on the significance
of the operations and the resulting transition probabilities,  
let us discuss next how one can derive from them some 
pertinent physical information. Quantities  of primary  
interest are the probabilities $\omega_A(E)$ of meeting in an 
ensemble, described by a state $\omega_A \in \bomegaN$,
members with specific properties, described by a projection 
\mbox{$E \in \bcN$}. As we will see,  these data 
can be derived from the transition probabilities for 
suitably tuned operations (reminiscent of tuned up detectors).

According to  part (iii) of the preceding lemma
one has for any given projection $E$ and 
isometry $V \in \bcN$ with range $E$
the \textit{a priori} bound $|\omega_A(V)| \leq \omega_A(E)$.
It then follows from~(ii) that  
for any given finite number of states
and $\varepsilon > 0$ there exist appropriate isometries   $V$ 
satisfying $|\omega_A(V) - \omega_A(E)| < \varepsilon$,
\ie they saturate the upper bound with arbitrary precision.  
Moreover, as an immediate consequence of part (i), there exist 
for any $\varepsilon > 0$  
unitaries $U \in  \bcN$ such that 
$\| \omega_{\, UA} - \omega_{\, VA} \| < \varepsilon $
for given isometry $V$ and states $\omega_A$.
Combining this  information  one finds that
$\omega_{\, U A} (1- E) < \varepsilon$ and
$|\omega_A \cdot \omega_{\, U A}  - \omega_A(E)^2 | < 4 \varepsilon $
for the given states. The corresponding operations~$\Ad{U}$ 
thus (a) determine within the given margins the probabilities 
$\omega_A(E) \approx (\omega_A \cdot \omega_{U A})^{1/2}$ 
of finding the property $E$ in the initial states  
and (b) create final states which 
have the property $E$  with probability $\omega_{U A}(E) \approx 1$ 
without relying on the standard process of (von Neumann--L\"uders) 
state reduction.
This construction can be extended to several
commuting projections $E_m$, $m = 1, \dots , M$.
One can then determine the mean values of 
observables $O = \sum_{m=1}^M  o_m E_m$ in the states 
$\omega_A$ with arbitrary precision 
by means of operations $\Ad{U_m}$, $m = 1, \dots , M$,
making use of the relation 
$\omega_A(O)  \approx  \sum_{m=1}^M o_m \, (\omega_A \cdot \omega_{U_m A})^{1/2}$.

Having seen how suitably tuned operations allow for the  
recovery of observables and their statistical interpretation, 
let us stress that there exist operations providing  
information which cannot be obtained in the conventional 
setting of observable algebras.  To illustrate this fact, 
let us recall that for any 
non-trivial projection $E \in \bcN$ one has 
$\omega(E) > 0$ for generic states~$\omega$,
the vacuum in quantum field theory being an example.
This Reeh--Schlieder property  
of the vacuum is often regarded as a conceptual problem \cite{Hal,BuYn}.
For it implies that there do not exist 
perfect detectors, described by projections 
$E$,  which give a non--zero
signal exclusively in states which are different from
the vacuum, \ie for which $\omega_A(E) \neq 0$ implies 
$\omega_A \neq \omega$. Yet, relying on the notion of 
operations and transition probabilities, one can model 
such detectors: picking any unitary $U \in \bcN$
with the property that $\omega(U) = 0$, the operation 
$\omega_A \mapsto \omega_A \circ \Ad{U} = \omega_{UA}$,
$\omega_A \in \bomegaN$, has the desired 
property since $\omega_A \cdot \omega_{UA} 
= |\omega_A(U)|^2 \neq 0$ implies $\omega_A \neq \omega$.
The quantity $\omega_A \cdot \omega_{UA}$ may therefore  
consistently be interpreted as the probability that the
detector registers some deviation from~$\omega$
in the state $\omega_A \in \bomegaN$. 

The upshot of this discussion is the insight that 
operations and the resulting transition 
probabilities provide a meaningful extension of the conventional 
framework of observable algebras and their statistical  
interpretation. We therefore propose to call 
these operations \textit{primitive observables},  
(primitive in the sense of being basic). 
Fundamental concepts which are familiar from the conventional 
setting, such as the
characterization of commensurable observables, can be transferred 
to the primitive observables in a meaningful manner. Two 
primitive observables, described by 
operations $\Ad{U_1}$ and $\Ad{U_2}$, are  
regarded as \mbox{commensurable}  
if the corresponding composed operations coincide,
$\Ad{\, U_1 U_2} = \Ad{\, U_2 U_1}$.
They then give rise to \mbox{coinciding} transition 
probabilities for all states, \ie 
\mbox{$|\omega_A(U_1 U_2)|^2 =|\omega_A(U_2 U_1)|^2$,} 
\mbox{$\omega_A \in \bomegaN$.} Note that this condition does 
not necessarily imply that the unitaries commute. 
Yet if one has two sequences of commuting primitive observables 
constructed from unitaries which are tuned as above so as to  
approximate projections $E_1, E_2$ with arbitrary 
precision, then one can 
show that these projections commute. So one recovers
the familiar commutativity of commensurable observables in this case. 
However, the generalized condition of commensurability imposes  
weaker constraints on the primitive observables.
It may therefore be useful in discussions  
of causality properties of theories, where the standard postulate of 
locality of observables \cite{Ha} cannot be applied. In any case, 
the proposed generalized notion of commensurability fits naturally into 
the present setting and is physically meaningful. 

Up to this point we have concentrated on properties of the 
states in $\bomegaN$ which do not depend on the type of the underlying
generic state $\omega$. Yet there is some
physically relevant feature of these states which does depend on their type.
It originates from the fact that the notion of transition 
probability acquires physical significance through the existence of 
tunable operations acting on ensembles. Within the theoretical
framework this fact is expressed 
by the assertion that for any $U \in \bcN$ the mapping 
$\omega_A \mapsto \omega_{\, UA}$ has some operational meaning and 
the corresponding transition probability \ $\omega_A \cdot  \omega_{UA}$
is in principle accessible to observations. 
One may therefore ask
whether it is possible to determine in this way the 
transition probabilities $\omega_A \cdot \omega_B$ between 
arbitrary states $\omega_A, \omega_B \in \bomegaN$.

An affirmative answer to this question 
requires that for any pair of states there exists some 
unitary operator $U \in \bcN$ such that $\omega_B = \omega_{UA}$ or,
slightly less restrictive, that there exist suitable unitaries     
for which this equality can be established with arbitrary precision.  
In other words, the inner operations induced by unitary 
operators have to act (topologically) transitively
on the states~$\bomegaN$. Otherwise, it would not be possible
to determine by physical operations
the transition probabilities for arbitrary pairs of
states. As a matter of fact, there do 
exist generic states~$\omega$ on $\bcN$ where 
this obstruction occurs. These are primary 
states of type~III$_\lambda$, $0 \leq \lambda < 1$, for which 
there exist certain pairs \ $\omega_A, \omega_B \in \bomegaN$ 
whose minimal distance $\inf_U \| \omega_B - \omega_{UA} \|$  
with regard to all possible operations is strictly positive~\cite{CoHaSt}.

There are, however, two important classes of states $\omega$ 
which comply with the  condition of (topologically) transitive 
action of  operations on the respective states~$\bomegaN$.
The first class consists of pure states~$\omega$ on~$\bcN$, which 
are known to satisfy this condition 
according to a well--known theorem by Kadison \cite{Ka}. 
The second class consists of states~$\omega$ on~$\bcN$ 
which are of type~III$_1$. They also comply with this condition, as has been 
shown in \cite{CoSt}. So these two classes of states have 
much more in common than one might infer from their 
mathematical definition based on modular theory. More remarkably, these two 
classes consist exactly of those states which abound in 
the context of infinite quantum systems \cite{Ha, BrRo, BuRo, Lo, Y}. 
Thus the concepts introduced here  
fully cover these classes of primary physical interest.

\section{Summary and outlook}
\setcounter{equation}{0}

In the present investigation we have established some 
universal properties of states in infinite quantum systems, 
which are described by funnels of type I$_\infty$  factors. 
These states may be regarded as excitations of some 
generic reference state, describing a global background in which
local measurements and operations take place. The states 
form a complete set 
with regard to local operations and therefore 
provide a meaningful framework for the discussion
of their physical properties. 
Even though the states can be of any 
infinite type, they share 
many basic properties with pure states. In particular, they
allow for an intrinsic definition of coherent 
superpositions and of transition probabilities. The physical 
interpretation of this novel framework is based on the concept
of primitive observables which extends the familiar 
notion of observables in terms of operator algebras. 
Primitive observables admit a consistent statistical 
interpretation, related to actual observations, and they bypass  
certain counter intuitive features and apparent paradoxa
of the conventional operator algebraic setting. Moreover, 
they comply with a generalized condition of 
commensurability which entails the standard commutativity
property of commensurable observables but is less
restrictive. It thereby allows to discuss  
causality properties of theories without relying
on the usual condition of locality. 

Whereas the present 
arguments rely on the existence of funnels
of type I$_\infty$ factors, one does not need to 
exhibit these funnels explicitly in applications 
of our results. Let us illustrate this fact
in relativistic quantum field theory, the simplest examples 
being theories of non--interacting 
particles. There one can construct for suitable nets of  
open bounded spacetime regions~$\cO \subset \RR^d$
corresponding von Neumann algebras $\cA(\cO)$ 
which are interpreted as algebras generated by 
observables localized in the respective regions. 
These algebras are known to be factors of type~III$_1$ and one also 
knows that they have the split property, \ie
for any proper inclusion of regions $\cO_1 \subset \cO_2$
there exists some type I$_\infty$ factor $\cN$ such that
\mbox{$\cA(\cO_1) \subset \cN \subset \cA(\cO_2)$}, 
c.f. \cite{Ha} and references quoted there. It follows
that for any sequence of properly increasing regions
$\cO_n$ there exist type I$_\infty$ factors $\cN_n$, $n \in \NN$,
generating a funnel $\bcN$, such that one has  
$\bcA \doteq \bigcup_{n \in \NN} \cA(\cO_n) = 
\bigcup_{n \in \NN} \cN_n = \bcN$. Hence, in spite of the fact  
that the intermediate type~I$_\infty$ factors are not
explicitly known, the present results also apply  
to the concretely given funnel $\bcA$ of type III$_1$ factors. 
This correspondence  has been established in many 
other field theoretic models, c.f. for example
\cite{DoLo2,Dr,Fr,DaLoRa,Le,Su,Ve}, and also by general  
structural arguments \cite{BuDaFr,Fe}. 
The present framework thus provides a promising 
setting for further study of the properties and the 
interpetation of states in these infinite quantum systems

On the mathematical side, the present results offer a 
new look at the state spaces of hyperfinite factors. 
They are all completions of ``skeleton spaces''  
of finite rank operators 
which are bi--modules relative to the funnel 
generating the respective factor. It is of interest in this
context that, algebraically, the funnels are all isomorphic
and that funnels generating a given factor
are related by inner automorphisms \cite{Ta}. 
Thus this description of the state spaces 
in terms of finite rank operators
is intrinsic. It seems an interesting problem to 
understand how the type
of the factors is encoded in the structure of their
respective  skeleton spaces.

\end{document}